\documentclass[12pt]{article}
\usepackage{amsmath, amsthm, amssymb, setspace, mathrsfs, color, verbatim, graphicx,hyperref, amsfonts, bm, rotating, lscape, multicol, multirow,bbm,natbib,Sweave} 
\usepackage[left=0.7in,right=1.0in,top=0.7in,bottom=1.0in,nohead]{geometry}
\usepackage[ruled,lined,boxed]{algorithm2e}

\newtheorem{thm}{Theorem}

\newtheorem{lem}[thm]{Lemma}

\def\pkg#1{\texttt{\textbf{#1}}}
\def\code#1{\texttt{#1}}
\def\proglang#1{\textbf{#1}}

\newcommand{\utwi}[1]{\mbox{\boldmath $ #1$}}
\newcommand{\Y}{{\utwi{Y}}}
\newcommand{\X}{{\utwi{X}}}
\providecommand{\e}[1]{\ensuremath{\times 10^{#1}}}

\begin{document}

\title{\pkg{ecp}: An \proglang{R} Package for Nonparametric Multiple Change Point Analysis of Multivariate Data}
\author{Nicholas A. James and David S. Matteson\\ Cornell University\footnote{
James is a PhD Candidate, 
School of Operations Research and Information Engineering,
Cornell University,
206 Rhodes Hall,
Ithaca, NY 14853
(Email: \href{mailto:nj89@cornell.edu}{nj89@cornell.edu}; Web: \url{https://courses.cit.cornell.edu/nj89/}).
Matteson is an Assistant Professor, 
Department of Statistical Science,
Cornell University,
1196 Comstock Hall,
Ithaca, NY 14853
(Email: \href{mailto:matteson@cornell.edu}{matteson@cornell.edu}; Web: \url{http://www.stat.cornell.edu/\~matteson/}).}
}
\date{}
\maketitle

\begin{abstract}
There are many different ways in which change point analysis can be performed, from purely parametric methods to those that are distribution 
free. The \pkg{ecp} package is designed to perform multiple change point analysis while making as few assumptions as possible. While many other 
change point methods are applicable only for univariate data, this \proglang{R} package is suitable for both univariate and multivariate observations. Hierarchical estimation 
can be based upon either a divisive or agglomerative algorithm. Divisive estimation sequentially identifies change points via a 
bisection algorithm. The agglomerative algorithm estimates change point locations by determining an optimal segmentation. Both approaches 
are able to detect \emph{any} type of distributional change within the data. This provides an advantage over many existing change point algorithms which 
are only able to detect changes within the marginal distributions.
\end{abstract}

\par\vfill\noindent
\par\medskip\noindent
{\bf Short title: \pkg{ecp}: An R Package for Nonparametric Multiple Change Point Analysis}

\clearpage\pagebreak\newpage
\newlength{\gnat}
\setlength{\gnat}{24pt} 
\baselineskip=\gnat

\section{Introduction}\label{introduction}
Change point analysis is the process of detecting distributional changes within time-ordered observations. This arises in financial modeling 
\Citep{Talih:2005}, where correlated assets are traded and models are based on historical data. It is applied in bioinformatics \Citep{Muggeo:2011} to 
identify genes that are associated with specific cancers and other diseases. Change point analysis is also used to detect credit card fraud 
\Citep{Bolton:2002} and other anomalies \Citep{Akoglu:2010,Sequeira:2002}; and for data classification in data mining \Citep{Mampaey:2011}.\par

We introduce the \pkg{ecp} \proglang R package for multiple 
change point analysis of multivariate time series \citep{James:2012}. The \pkg{ecp} package provides methods for change point analysis 
that are able to detect \emph{any} type of distributional change within a time series. Determination of the number of change points is also addressed 
by these methods as they estimate both the number and locations of change points simultaneously. The only assumptions placed on distributions are that the absolute 
$\alpha$th moment exists, for some $\alpha\in (0,2]$, and that observations are independent over time. Distributional changes are identified by making use of the 
energy statistic of \cite{Rizzo:2005, Rizzo:2010}.\par

There are a number of freely available \proglang R packages that can be used to perform change point analysis, each making its own assumptions 
about the observed time series. For instance, the \pkg{changepoint} package \citep{Killick:2011} provides many methods for performing change 
point analysis of univariate time series. Although the package only considers the case of independent observations, the theory behind the implemented methods 
allows for certain types of serial dependence \citep{Killick:2012}. For specific methods, the expected computational cost can 
be shown to be linear with respect to the length of the time series. Currently, the
\pkg{changepoint} package is only suitable for finding changes in mean or variance. This package also estimates multiple change points through the 
use of penalization. The drawback to this approach is that it requires a user specified penalty term.\par

The \pkg{cpm} package \citep{Ross:2012} similarly provides a variety of methods for performing change point analysis 
of univariate time series. These methods range from those to detect changes in independent Gaussian data to fully nonparametric methods 
that can detect general distributional changes. Although this package provides methods to perform analysis of univariate time series with arbitrary 
distributions, these methods cannot be easily extended to detect changes in the full joint distribution of multivariate data.\par

Unlike the \pkg{changepoint} and \pkg{cpm} packages, the \pkg{bcp} package \citep{Emerson:2007} is designed to perform Bayesian single change 
point analysis of univariate time series. It returns the posterior probability of a change point occurring 
at each time index in the series. Recent versions of this package have reduced the computational cost from quadratic to linear with respect 
to the length of the series. However, all versions of this package are only designed to detect changes in the mean of independent Gaussian 
observations.\par

The \pkg{strucchange} package \citep{strucchangepkg} provides a suite of tools for detecting changes within linear regression models. Many of these tools however, 
focus on detecting at most one change within the regression model. This package also contains methods that perform online change detection, thus allowing it to be 
used in settings where there are multiple changes. Additionally, if the number of changes is known \emph{a priori} then the \code{breakpoints} method \citep{breakpoints} 
can be used to perform retrospective analysis. For a given number of changes, this method returns the change point estimates which minimize the residual sum of squares.\par

In Section~\ref{energy-statistic} we introduce the energy statistic of \cite{Rizzo:2005,
Rizzo:2010}, which is the fundamental divergence measure applied for change point analysis. Sections~\ref{e-divisive} and~\ref{e-agglomerative} 
briefly outline the package's methods. Section~\ref{Examples} provide examples of these methods being applied to simulated data, while 
Section~\ref{RealData} presents applications to real datasets. In the Appendix we include 
an outline of the algorithms used by this package's methods. Finally, the \pkg{ecp} package can be freely obtained at \url{http://cran.r-project.org/web/packages/ecp/}.\par

\section{The ecp package}\label{energy-statistic}
The \pkg{ecp} package is designed to address many of the limitations of the currently available change point packages. It  
is able to perform multiple change point analysis for both univariate and multivariate time series. The methods are able to estimate 
multiple change point locations, without {\it a priori} knowledge of the number of change points. The procedures assume that observations are 
independent with finite $\alpha$th absolute moments, for some $\alpha\in(0,2]$.

\subsection{Measuring differences in multivariate distributions}\label{population-energy}
\cite{Rizzo:2005, Rizzo:2010} introduce a divergence measure that can determine whether two independent random vectors are identically distributed. 
Suppose that $X,Y\in\mathbb R^d$ are such that, $X\sim F$ and $Y\sim G$, with characteristic functions $\phi_x(t)$ and $\phi_y(t)$, respectively. A 
divergence measure between the two distributions may be defined as
$$\int_{\mathbb R^d}|\phi_x(t)-\phi_y(t)|^2\;w(t)\,dt,$$
in which $w(t)$ is any positive weight function, for which the above integral is defined. Following \cite{James:2012} we employ the following weight 
function, 
$$w(t;\alpha)=\left(\frac{2\pi^{d/2}\Gamma(1-\alpha/2)}{\alpha 2^\alpha\Gamma[(d+\alpha)/2]}|t|^{d+\alpha}\right)^{-1},$$
for some fixed constant $\alpha\in(0,2)$. Thus our divergence measure is
$$\mathcal D(X,Y;\alpha)=\int_{\mathbb R^d}|\phi_x(t)-\phi_y(t)|^2\!\left(\frac{2\pi^{d/2}\Gamma(1-\alpha/2)}{\alpha 2^\alpha\Gamma[(d+\alpha)/2]}
|t|^{d+\alpha}\right)^{-1}\!\!dt.$$
An alternative divergence measure based on Euclidean distances may be defined as follows 
\begin{equation*}\label{eDist}
\mathcal E(X,Y;\alpha)=2E|X-Y|^\alpha-E|X-X'|^\alpha-E|Y-Y'|^\alpha.
\end{equation*}
In the above equation, $X'$ and $Y'$ are independent copies of $X$ and $Y$, respectively. Then given our choice of weight function, we have the 
following result.
\begin{lem}\label{thm1}
For any pair of independent random variables $X,Y\in\mathbb R^d$ and for any $\alpha\in(0,2)$, if $E(|X|^\alpha+|Y|^\alpha)<\infty$, then 
$\mathcal D(X,Y;\alpha)=\mathcal E(X,Y;\alpha)$, $\mathcal E(X,Y;\alpha)\in[0,\infty)$, and $\mathcal E(X,Y;\alpha)=0$ if and only if $X$ and $Y$ are 
identically distributed.
\end{lem}
\begin{proof}
A proof is given in the appendices of \cite{Rizzo:2005} and \cite{James:2012}.
\end{proof}
Thus far we have assumed that $\alpha\in(0,2)$, because in this setting $\mathcal E(X,Y;\alpha)=0$ if and only if $X$ and $Y$ are identically 
distributed. However, if we allow for $\alpha=2$ a weaker result of equality in mean is obtained.
\begin{lem}\label{thm2}
For any pair of independent random variables $X,Y\in\mathbb R^d$, if $E(|X|^2+|Y|^2)<\infty$, then $\mathcal D(X,Y;2)=\mathcal E(X,Y;2)$, $\mathcal 
E(X,Y;2)\in[0,\infty)$, and $\mathcal E(X,Y;2)=0$ if and only if $E X=E Y$.
\end{lem}
\begin{proof}
See \cite{Rizzo:2005}.
\end{proof}

\subsection{A sample divergence for multivariate distributions}\label{sample-energy}
Let $X\sim F$ and $Y\sim G$ for arbitrary distributions $F$ and $G$. Additionally, select $\alpha\in(0,2)$ such that $E|X|^\alpha$, $E|Y|^\alpha<\infty$. 
Let $\X_n=\{X_i: i=1, 2,\dots, n\}$ be $n$ independent observations with $X_i\sim F$, and $\Y_m=\{Y_j: j=1,\dots, m\}$ are 
$m$ independent observations with $Y_j\sim G$. Furthermore, we assume full mutual independence between all observations, $\X_n\perp\!\!\!\perp \Y_m$. 
Then Lemmas \ref{thm1} and \ref{thm2} suggest following sample divergence measure,
\begin{equation}\label{eStat}
\widehat{\mathcal E}(\X_n,\Y_m;\alpha)=\frac{2}{mn}\sum_{i=1}^n\sum_{j=1}^m|X_i-Y_j|^\alpha-\binom{n}{2}^{-1}\!\!\!\!\!\!\!\sum_{1\le i<k\le n}\!\!\!\!\!|X_i-X_k|^\alpha-
\binom{m}{2}^{-1}\!\!\!\!\!\!\!\sum_{1\le j<k\le m}\!\!\!\!\!|Y_j-Y_k|^\alpha.
\end{equation}
By the strong law of large numbers for $U$-statistics \citep{Hoeffding:1961} $\widehat{\cal E}(\X_n,\Y_m;\alpha)\stackrel{a.s.}{\to}{\cal E}(X,Y;\alpha)$ 
as $n\wedge m\to\infty$. Equation~\ref{eStat} allows for an estimate of $\mathcal D(X,Y;\alpha)$ without performing $d$-dimensional 
integration. Furthermore, let
\begin{equation*}
\widehat{\cal Q}(\X_n,\Y_m;\alpha)=\frac{mn}{m+n}\widehat{\cal E}(\X_n,\Y_m;\alpha)	
\end{equation*}
denote the scaled empirical divergence. Under the null hypothesis of equal distributions, i.e., $\mathcal E(X,Y;\alpha)=0$, \cite{Rizzo:2010} show that 
$\widehat{\mathcal Q}(\X_n,\Y_m;\alpha)$ converges in distribution to a non-degenerate random variable $\mathcal Q(X,Y;\alpha)$ as $m\wedge n\to\infty$. 
Specifically,
$$\mathcal Q(X,Y;\alpha)=\sum_{i=1}^\infty\lambda_i Q_i$$
in which the $\lambda_i\ge 0$ are constants that depend on $\alpha$ and the distributions of $X$ and $Y$, and the $Q_i$ are iid chi-squared random variables 
with one degree of freedom. Under the alternative hypothesis of unequal distributions, i.e., $\mathcal E(X,Y;\alpha)>0$, 
$\widehat{\mathcal Q}(\X_n,\Y_m;\alpha)\to\infty$ almost surely as $m\wedge n\to\infty$.\par

\cite{gretton:2007} use a statistic similar to that presented in Equation~\ref{eStat} to test for equality in distribution. This alternative statistic is used by 
\cite{schauer:2010} to test whether the point set inside a new perturbed set of cells had a different distribution to the ``unperturbed'' cells. However, this statistic, unlike the 
one presented in Equation~\ref{eStat}, is only applicable when the unknown distributions have continuous bounded density functions.\par

Using these facts we are able to develop two hierarchical methods for performing change point analysis, which we present in Sections~\ref{e-divisive} and~\ref{e-agglomerative}.

\section{Hierarchical divisive estimation: E-Divisive}\label{e-divisive}
We first present the E-Divisive method for performing hierarchical divisive estimation of multiple change points. Here, multiple change points are estimated 
by iteratively applying a procedure for locating a single change point. At each iteration a new change point location is estimated so that it divides 
an existing segment. As a result, the progression of this method can be diagrammed as a binary tree. In this tree, the root node corresponds to the case 
of no change points, and thus contains the entire time series. All other non-root nodes are either a copy of their parent, or correspond to one of the new 
segments created by the addition of a change point to their parent. Details on the estimation of change point locations can be found in \cite{James:2012}.\par

The statistical significance of an estimated change point is determined through a permutation test, since the distribution of the test statistic depends 
upon the distributions of the observations, which is \emph{unknonwn} in general. Suppose that at the $k$th iteration the current set of change points has segmented the time series in 
the $k$ segments $S_1, S_2,\dots, S_k$, and that we have estimated the next change point location as $\hat\tau_k$, which has an associated test statistic value 
of $q_0$. We then obtain our permuted sample by permuting the observations within each of $S_1,\dots,S_k$. Then, conditional on the previously estimated change point 
locations, we estimate the location of the next change point in our permuted sample, $\hat\tau_{k,r}$, along with its associated testing statistic value $q_r$. Our 
approximate $p$~value is then calculated as $\hat p=\#\{r: q_r\ge q_0\}/(R+1)$, where $R$ is the total number of permutations performed.

The signature of the method used to perform analysis based on this divisive approach is
\begin{center}
	\code{e.divisive(X, sig.lvl = 0.05, R = 199, eps = 1e-3, 
		   half = 1000, k = NULL, min.size = 30, alpha = 1)}
\end{center}

Descriptions for all function arguments can be found in the package's help files. 
The time complexity of this method is $\mathcal{O}(kT^2)$, where $k$ is the number of estimated change points, and $T$ is the number of observations 
in the series. Due to the running time being quadratic in the length of the time series this procedure is not recommended for series larger than several 
thousand observations. A reduction in the required computation time can be achieved by adjusting the \code{eps} and \code{half} function arguments. 
These arguments are used to obtain a permutation $p$~value with a uniformly bounded resampling risk \citep{Gandy:2009}.\par

In the case of independent observations, \cite{James:2012} show that this procedure generates strongly consistent 
change point estimates. There are other faster approaches for performing nonparametric multiple change point analysis, however they do not have 
a similar consistency guarantee, which is why we recommend our divisive approach when appropriate. A more complete outline of the divisive algorithm 
is detailed in the Appendix.\par

\section{Hierarchical agglomerative estimation: E-Agglo}\label{e-agglomerative}
We now present the E-Agglo method for performing hierarchical agglomerative estimation of multiple change points. This method requires that an initial 
segmentation of the data be provided. This initial segmentation can help to reduce the computational time of the procedure. It also allows for the 
inclusion of {\it a priori} knowledge of possible change point locations, however if no such assumptions are made, then each observation can be assigned 
to its own segment. Neighboring segments are then sequentially merged to maximize a goodness-of-fit statistic. The estimated 
change point locations are determined by the iteration which maximized the penalized goodness-of-fit statistic. When using the E-Agglo procedure 
it is assumed that there is at least one change point present within the time series.\par

The goodness-of-fit statistic used in \cite{James:2012} is the between-within distance \citep{Rizzo:2005} among adjacent segments. Let 
$\mathcal C=\{C_1,\dots,C_n\}$ be a segmentation of the $T$ observations into $n$ segments. The goodness-of-fit statistic is defined as 
\begin{equation}\label{gof}
\widehat{\mathcal S}_n(\mathcal C;\alpha)=\sum_{i=1}^{n}\widehat{\cal Q}(C_i,C_{i+1};\alpha).
\end{equation}

Since calculating the true 
maximum of the goodness-of-fit statistic for a given initial segmentation would be too computationally intensive, a greedy algorithm is used to find an approximate 
solution. For a detailed explanation of how this algorithm is efficiently carried out see the Appendix.

If overfitting is a concern, it is possible to penalize the sequence of goodness-of-fit statistics. This is accomplished through the use of the \code{penalty} argument, 
which generates a penalty based upon change point locations. Thus, the change point locations are estimated by maximizing
$$\widetilde{\cal S}_k=\widehat{\cal S}_k+\mbox{\code{penalty}}(\vec\tau(k))$$
where $\vec\tau(k)=\{\tau_1,\tau_2,\dots,\tau_k\}$ is the set of change points associated with the goodness-of-fit statistic $\widehat{\cal S}_k$. Examples of penalty terms include
\begin{center}
\code{penalty1 <- function(cp)\{-length(cp)\}\\
penalty2 <- function(cp)\{mean(diff(sort(cp)))\}\\
}
\end{center}
Here \code{penalty1} corresponds to the function $\mbox{\code{penalty}}(\vec\tau(k))=-k$ while \code{penalty2} corresponds to the function $\displaystyle \mbox{\code{penalty}}(\vec\tau(k))=
\frac{1}{k+1}\sum_{i=1}^{k-1}\left[\tau_{i+1}-\tau_i\right]$. Both penalties favor segmentations with larger sizes. However, \code{penalty1} equally penalizes every additional 
change point, while \code{penalty2} takes the size of the new segments into consideration.

The signature of the method used to perform agglomerative analysis is
\begin{center}
	\code{e.agglo(X, member = 1:nrow(X), alpha = 1, penalty = function(cp)\{0\})\\}
\end{center}
Descriptions for all function arguments can be found in the package's help files. 
Like the E-Divisive method, this is quadratic in the number of observations with computational complexity $\mathcal O(T^2)$, 
however its complexity does not depend on the number of estimated change points.\par

\section{Examples}\label{Examples}
In this section we illustrate the use of both the \code{e.divisive} and \code{e.agglo} functions to perform multivariate change point analysis.

\subsection{Change in univariate normal distribution}\label{uninorm}
We begin with the simple case of identifying change in univariate normal distributions. For this we sequentially generate 100 independent samples from the following 
normal distributions: $\mathcal N(0,1), \mathcal N(0,\sqrt{3}), \mathcal N(2,1),$ and $\mathcal N(2,2)$.\par

\begin{Schunk}
\begin{Sinput}
R> set.seed(250)
R> library("ecp")
R> period1 <- rnorm(100)
R> period2 <- rnorm(100,0,3)
R> period3 <- rnorm(100,2,1)
R> period4 <- rnorm(100,2,4)
R> Xnorm <- matrix(c(period1,period2,period3,period4),ncol=1)
R> output1 <- e.divisive(Xnorm, R = 499, alpha = 1)
R> output2 <- e.divisive(Xnorm, R = 499, alpha = 2)
R> output2$estimates
\end{Sinput}
\begin{Soutput}
[1]   1 201 358 401
\end{Soutput}
\begin{Sinput}
R> output1$k.hat
\end{Sinput}
\begin{Soutput}
[1] 4
\end{Soutput}
\begin{Sinput}
R> output1$order.found
\end{Sinput}
\begin{Soutput}
[1]   1 401 201 308 108
\end{Soutput}
\begin{Sinput}
R> output1$estimates
\end{Sinput}
\begin{Soutput}
[1]   1 108 201 308 401
\end{Soutput}
\begin{Sinput}
R> output1$considered.last
\end{Sinput}
\begin{Soutput}
[1] 358
\end{Soutput}
\begin{Sinput}
R> output1$p.values
\end{Sinput}
\begin{Soutput}
[1] 0.002 0.002 0.010 1.000
\end{Soutput}
\begin{Sinput}
R> output1$permutations
\end{Sinput}
\begin{Soutput}
[1] 499 499 499   5
\end{Soutput}
\begin{Sinput}
R> ts.plot(Xnorm,ylab='Value',main='Change in a Univariate Gaussian Sequence')
R> abline(v=c(101,201,301),col='blue')
R> abline(v=output1$estimates[c(-1,-5)],col='red',lty=2)
\end{Sinput}
\end{Schunk}
As can be seen, if $\alpha=2$ the E-Divisive method can only identify changes in mean. For this reason, it is recommended that $\alpha$ is selected so as to lie 
in the interval $(0,2)$, in general. Figure~\ref{ex1} depicts the example time series, along with the estimated change points from the E-Divisive method when using $\alpha=1$.\par

Furthermore, when applying the E-Agglo method to this same simulated dataset, we obtain similar results. In the \proglang{R} code below we present the case for $\alpha=1$.\par

\begin{Schunk}
\begin{Sinput}
R> library("ecp")
R> member <- rep(1:40,rep(10,40))
R> output <- e.agglo(X = Xnorm, member = member, alpha = 1)
R> output$opt
\end{Sinput}
\begin{Soutput}
[1]   1 101 201 301 401
\end{Soutput}
\begin{Sinput}
R> tail(output$fit,5)
\end{Sinput}
\begin{Soutput}
[1] 100.05695 107.82542 104.30608 102.64330 -17.10722
\end{Soutput}
\begin{Sinput}
R> output$progression[1,1:10]
\end{Sinput}
\begin{Soutput}
[1]  1 11 21 31 41 51 61 71 81 91
\end{Soutput}
\begin{Sinput}
R> output$merged[1:4,]
\end{Sinput}
\begin{Soutput}
     [,1] [,2]
[1,]  -39  -40
[2,]   -1   -2
[3,]  -38    1
[4,]    2   -3
\end{Soutput}
\end{Schunk}

\begin{figure}[!ht]
	\centerline{\includegraphics{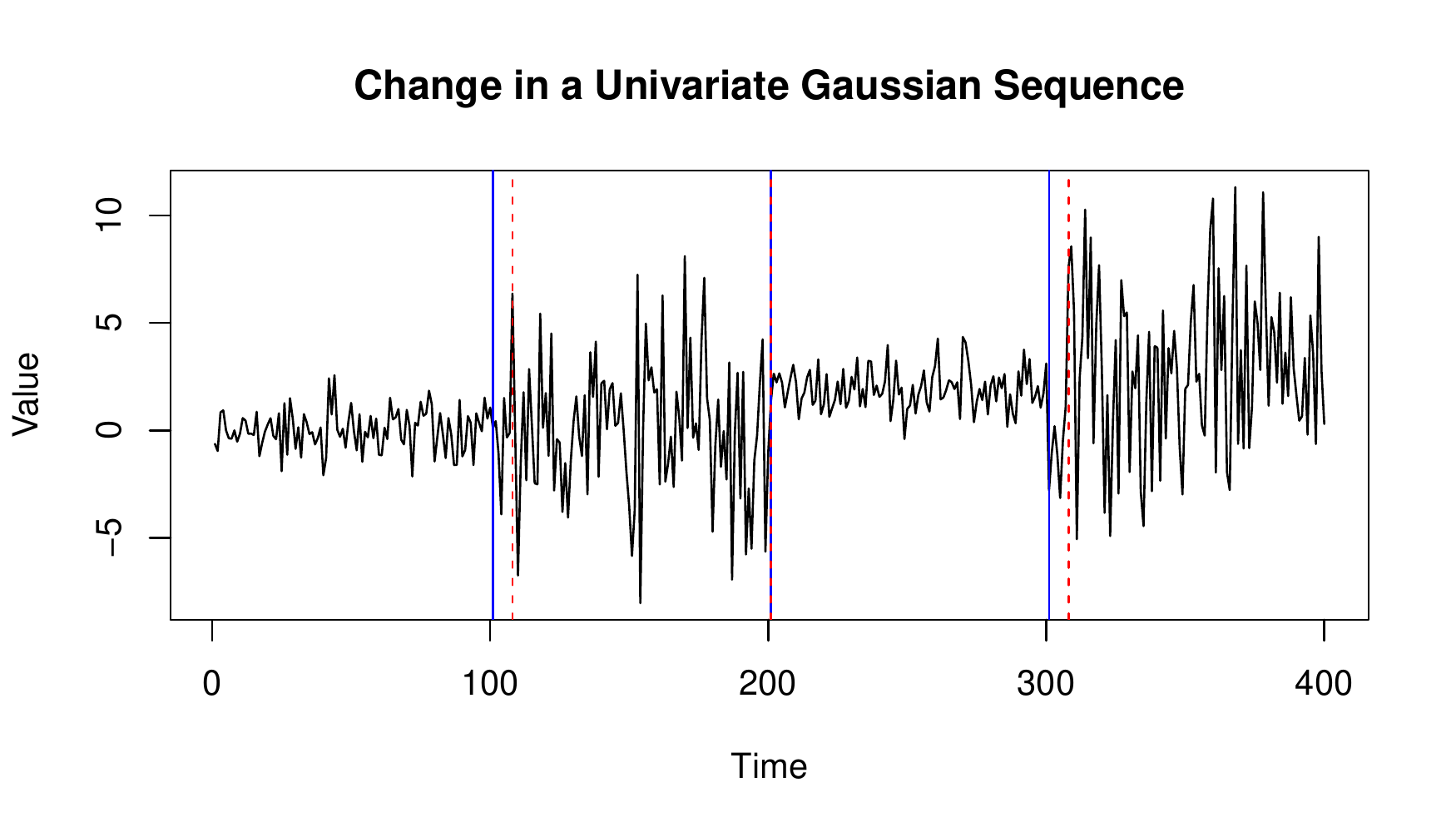}}
	\caption{Simulated independent Gaussian observations with changes in mean or variance. Dashed vertical lines 
indicate the change point locations estimated by the E-Divisive method, when using $\alpha=1$. Solid vertical lines indicate the true change point locations.}
	\label{ex1}
\end{figure}

\subsection{Multivariate change in covariance}\label{multicov}
To demonstrate that our methods do not just identify changes in marginal distributions we consider a multivariate example with only a change in covariance. In this example 
the marginal distributions remain the same, but the joint distribution changes. Therefore, applying a univariate change point procedure to each margin, such as those 
implemented by the \pkg{chagenpoint}, \pkg{cmp}, and \pkg{bcp} packages, will not detect the changes. The observations in this example are drawn from trivariate normal 
distributions with mean vector $\mu=(0,0,0)^\top$ and the following covariance matrices:
$$\begin{pmatrix}
      1&0&0\\
      0&1&0\\
      0&0&1
  \end{pmatrix}, 
\begin{pmatrix}
    1&0.9&0.9\\
      0.9&1&0.9\\
      0.9&0.9&1
\end{pmatrix},
\mbox{ and}
\begin{pmatrix}
    1&0&0\\
	0&1&0\\
    0&0&1
\end{pmatrix}.$$
Observations are generated by using the \pkg{mvtnorm} package \citep{Genz:2012}.\par

\begin{Schunk}
\begin{Sinput}
R> set.seed(200)
R> library("ecp")
R> library("mvtnorm")
R> mu <- rep(0,3)
R> covA <- matrix(c(1,0,0,0,1,0,0,0,1),3,3)
R> covB <- matrix(c(1,0.9,0.9,0.9,1,0.9,0.9,0.9,1),3,3)
R> period1 <- rmvnorm(250, mu, covA)
R> period2 <- rmvnorm(250, mu, covB)
R> period3 <- rmvnorm(250, mu, covA)
R> Xcov <- rbind(period1, period2, period3)
R> DivOutput <- e.divisive(Xcov, R = 499, alpha = 1)
R> DivOutput$estimates
\end{Sinput}
\begin{Soutput}
[1]   1 250 502 751
\end{Soutput}
\begin{Sinput}
R> member <- rep(1:15,rep(50,15))
R> pen = function(x){-length(x)}
R> AggOutput1 <- e.agglo(X = Xcov, member = member, alpha = 1)
R> AggOutput2 <- e.agglo(X = Xcov, member = member, alpha = 1, penalty = pen)
R> AggOutput1$opt
\end{Sinput}
\begin{Soutput}
[1] 1 101 201 301 351 501 601 701 751
\end{Soutput}
\begin{Sinput}
R> AggOutput2$opt
\end{Sinput}
\begin{Soutput}
[1] 301 501
\end{Soutput}
\end{Schunk}
In this case, the default procedure generates too many change points, as can be seen by the result of \code{AggOutput1}. When 
penalizing based upon the number of change points we obtain a much more accurate result, as shown by \code{AggOutput2}. Here the 
E-Agglo method has indicated that observations 1 through 300 and observations 501 through 750 are identically distributed.

\subsection{Multivariate change in tails}\label{tailsEx}
For our second multivariate example we consider the case where the change in distribution is caused by a change in tail behavior. Data points are 
drawn from a bivariate normal distribution and a bivariate t-distribution with 2 degrees of freedom. Figure~\ref{tails} depicts the different samples within the 
time series.\par

\begin{Schunk}
\begin{Sinput}
R> set.seed(100)
R> library("ecp")
R> library("mvtnorm")
R> mu <- rep(0,2)
R> period1 <- rmvnorm(250, mu, diag(2))
R> period2 <- rmvt(250, sigma = diag(2), df = 2)
R> period3 <- rmvnorm(250, mu, diag(2))
R> Xtail <- rbind(period1, period2, period3)
R> output <- e.divisive(Xtail, R = 499, alpha = 1)
R> output$estimates
\end{Sinput}
\begin{Soutput}
[1]   1 257 504 751
\end{Soutput}
\end{Schunk}

\begin{figure}[!ht]
	\centerline{\includegraphics{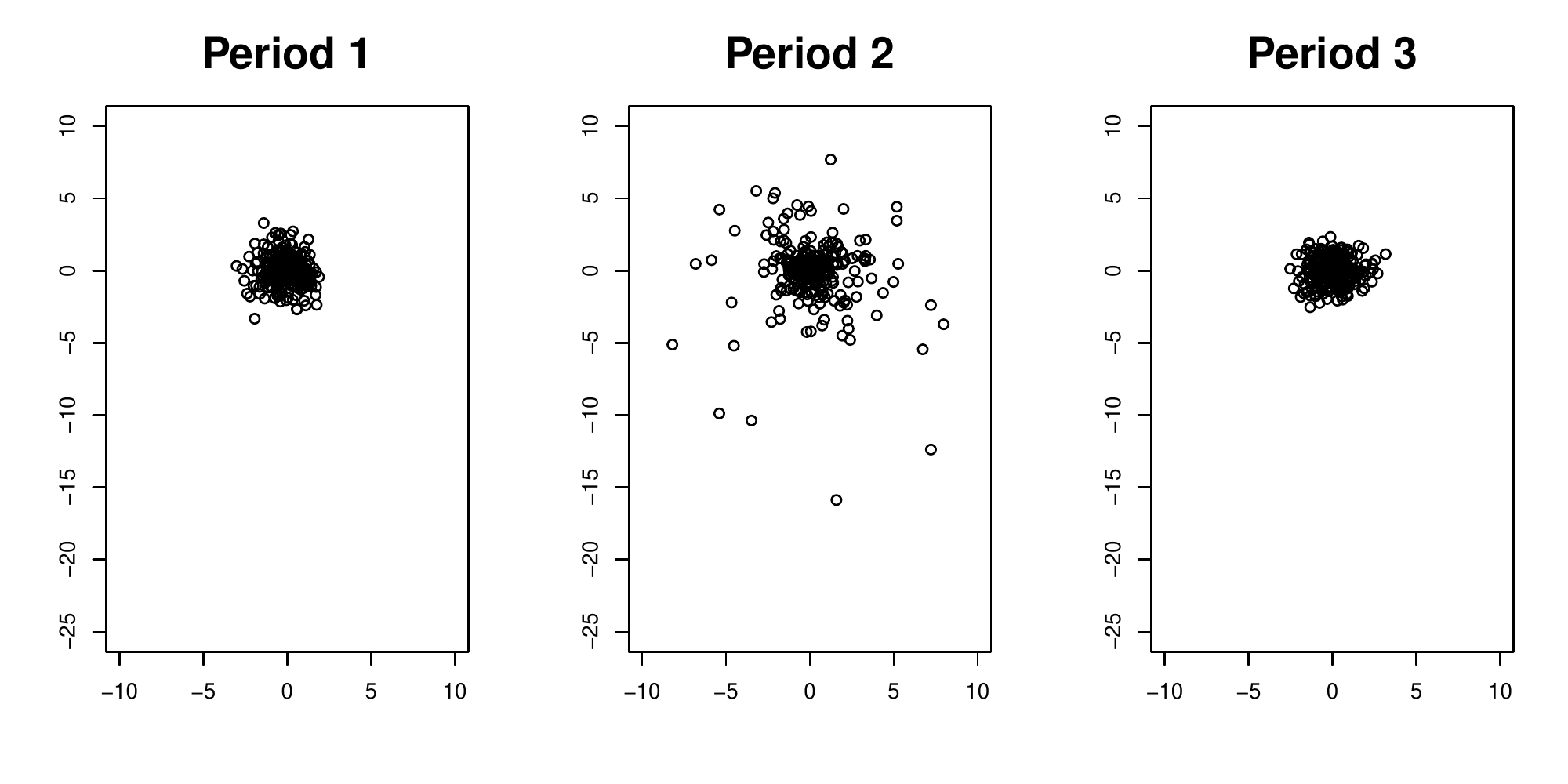}}
	\caption{Data set used for the change in tail behavior example from Section~\ref{tailsEx}. Periods 1 and 3 contain independent bivariate Gaussian observations 
with mean vector $(0,0)^\top$ and identity covariance matrix. The second time period contains independent observations from a bivariate Student's t-distribution 
with 2 degrees of freedom and identity covariance matrix.}
	\label{tails}
\end{figure}

\subsection{Inhomogeneous spatio-temporal point process}\label{stpp}
We apply the E-Agglo procedure to a spatio-temporal point process. The examined dataset consist of 10,498 observations, each with 
associated time and spatial coordinates. This dataset spans the time interval $[0,7]$ and has spatial domain $\mathbb R^2$. It contains 
3 change points, which occur at times $t_1=1, t_2=3,$ and $t_3=4.5$. Over each of these subintervals, $t\in[t_i,t_{i+1}]$ the 
process is an inhomogeneous Poisson point process with intensity function $\lambda(s,t)=f_i(s)$, a 2-d density function, for $i=1,2,3,4$. 
This intensity function is chosen to be 
the density function from a mixture of 3 bivariate normal distributions, 

$${\mathcal N}\left(\begin{pmatrix}-7\\-7\end{pmatrix},\begin{pmatrix}25&0\\0&25\end{pmatrix}\right),
\ \ {\mathcal N}\left(\begin{pmatrix}0\\0\end{pmatrix},\begin{pmatrix}9&0\\0&1\end{pmatrix}\right),
\ \mbox{and}\ \ {\mathcal N}\left(\begin{pmatrix}5.5\\0\end{pmatrix},\begin{pmatrix}9&0.9\\0.9&9\end{pmatrix}\right).$$
\noindent For the time periods, $[0,1], (1,3], (3,4.5],$ and $(4.5,7]$ the respective mixture weights are
$$\left(\frac{1}{3},\frac{1}{3},\frac{1}{3}\right),\ \left(\frac{1}{5},\frac{1}{2},\frac{3}{10}\right),
\ \left(\frac{7}{20},\frac{3}{10},\frac{7}{20}\right),\ \mbox{and}\ \left(\frac{1}{5},\frac{3}{10},\frac{1}{2}\right).$$

To apply the E-Agglo procedure we initially segment the observations into 50 segments such that each segment spans an equal amount 
of time. At its termination, the E-Agglo procedure, with no penalty, identified change points at times 0.998, 3.000, and 4.499. These results can be obtained 
with the following

\begin{Schunk}
\begin{Sinput}
R> library("mvtnorm"); library("combinat"); library("MASS"); library("ecp")
R> set.seed(2013)
R>
R> lambda <- 1500  # This is the overall arrival rate per unit time. 
R> #set of distribution means
R> muA <- c(-7,-7); muB <- c(0,0); muC <- c(5.5,0)
R> #set of distribution covariance matrices
R> covA <- 25*diag(2)
R> covB <- matrix(c(9,0,0,1),2)
R> covC <- matrix(c(9,.9,.9,9),2)
R> #time intervals
R> time.interval <- matrix(c(0,1,3,4.5,1,3,4.5,7),4,2)
R> #mixing coefficents
R> mixing.coef <- rbind(c(1/3,1/3,1/3),c(.2,.5,.3),c(.35,.3,.35),c(.2,.3,.5))
R> 
R> stppData <- NULL
R> for(i in 1:4){
+  count <- rpois(1, lambda* diff(time.interval[i,]))
+  Z <- rmultz2(n = count, p = mixing.coef[i,])
+  S <- rbind(rmvnorm(Z[1],muA,covA), rmvnorm(Z[2],muB,covB),
+	    rmvnorm(Z[3],muC,covC))
+  X <- cbind(rep(i,count), runif(n = count, time.interval[i,1],
+	    time.interval[i,2]), S)
+  stppData <- rbind(stppData, X[order(X[,2]),])
+ }
R> 
R> member <- as.numeric(cut(stppData[,2], breaks = seq(0,7,by=1/12)))
R> output <- e.agglo(X = stppData[,3:4], member = member, alpha = 1)
\end{Sinput}
\end{Schunk}
The E-Agglo procedure was also run on the above data set using the following penalty function,
\begin{itemize}
	\item \code{pen <- function(cp)\{ -length(cp) \} }
\end{itemize}
When using \code{pen}, change points were also estimated at times 0.998, 3.000, 4.499 
The progression of the goodness-of-fit statistic for the different schemes is plotted in Figure~\ref{gofStpp}. A 
comparison of the true densities and the estimated densities obtained from the procedure's results with no penalty are shown in Figures~\ref{trueStpp} 
and ~\ref{estStpp}, respectively. As can be see, the estimated results obtained from the E-Agglo procedure provide a reasonable approximation to the 
true densities.\par

\begin{figure}[!ht]
	\centerline{\includegraphics{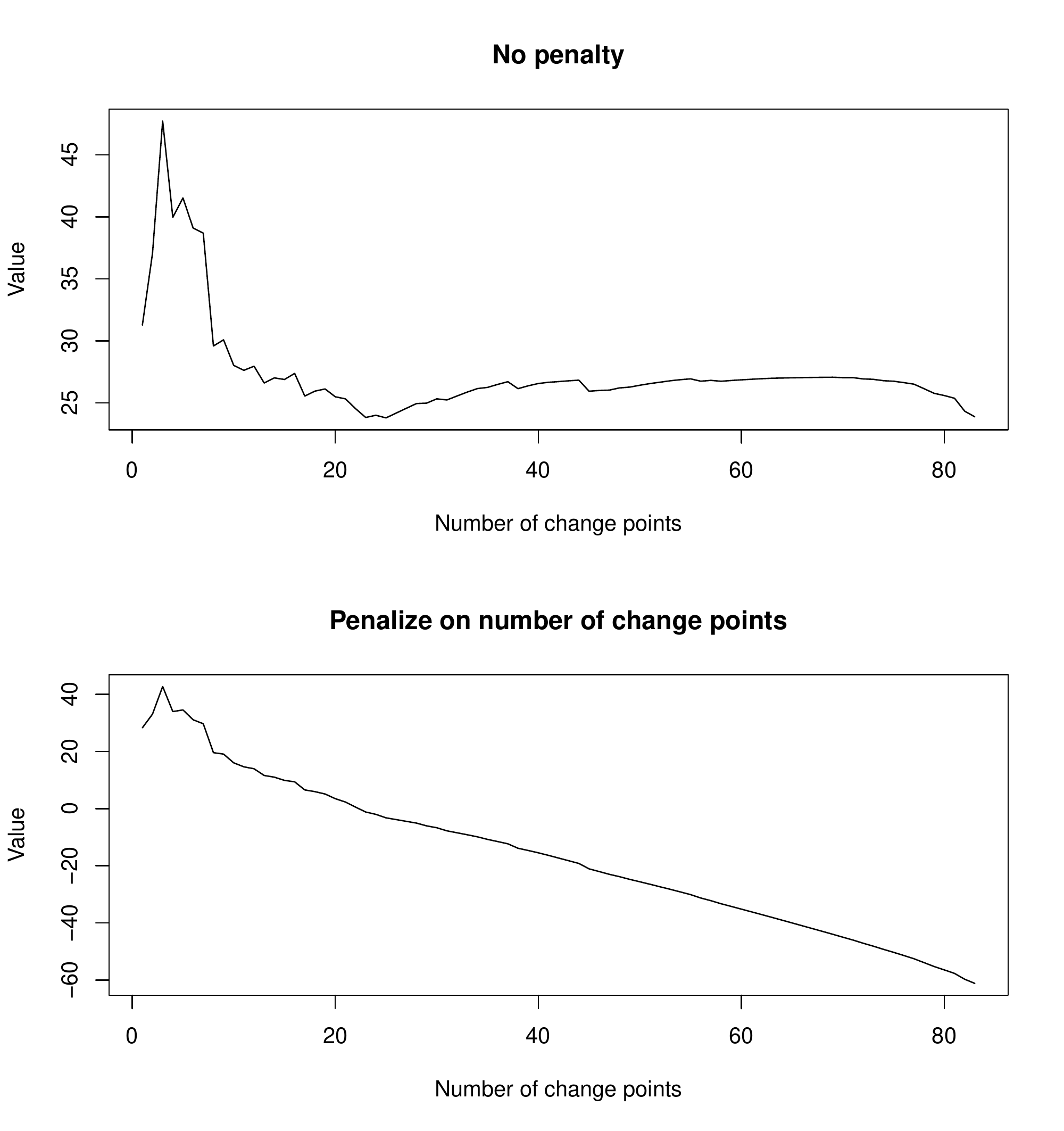}}
	\caption{The progression of the goodness-of-fit statistic for the various penalization schemes discussed in Section~\ref{stpp}.}
	\label{gofStpp}
\end{figure}

\begin{figure}[!ht]
	\centerline{\includegraphics{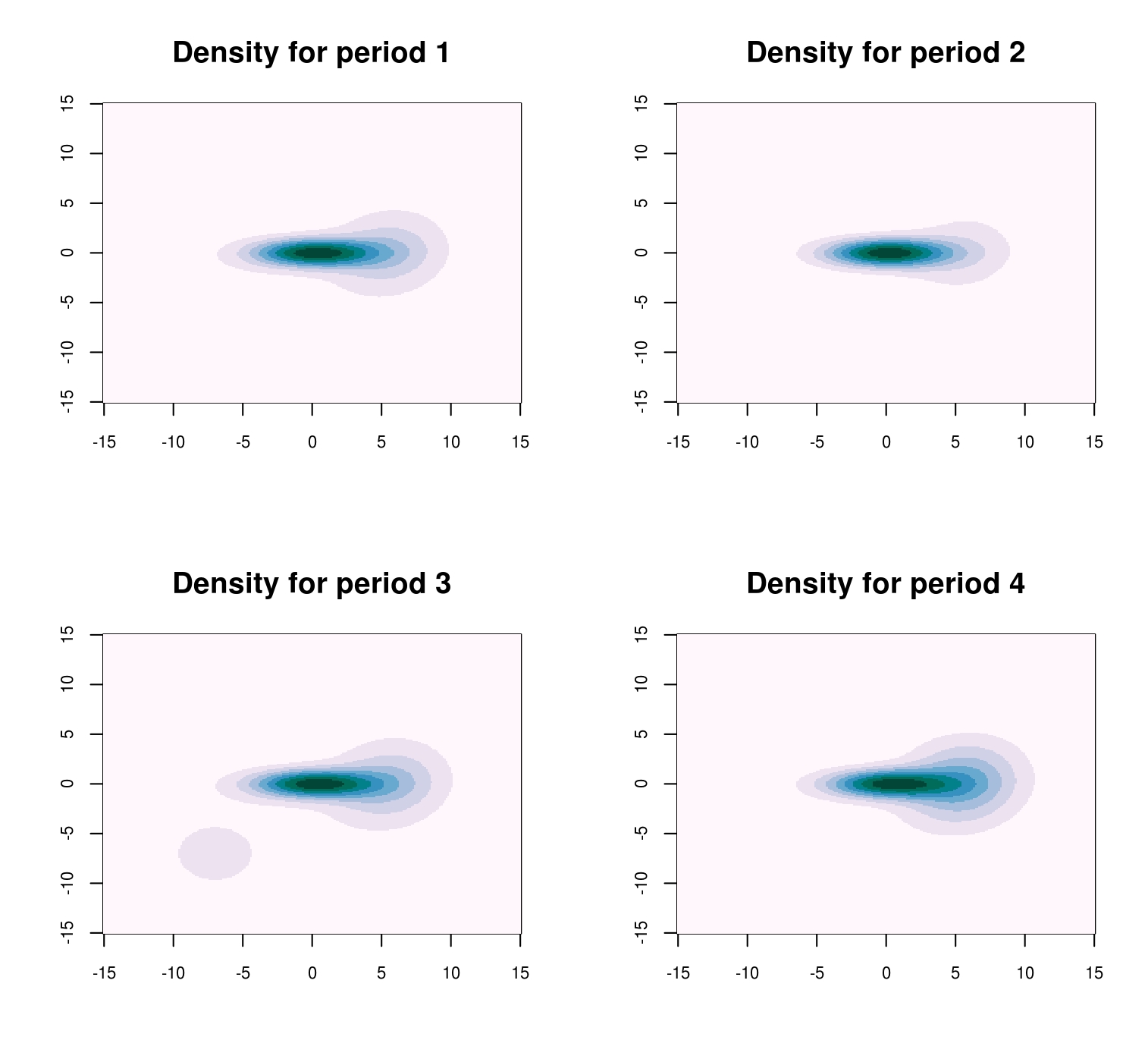}}
	\caption{True density plots for the different segments of the spatio-temporal point process in Section~\ref{stpp}.}
	\label{trueStpp}
\end{figure}

\begin{center}
\begin{figure}[!ht]
	\centerline{\includegraphics{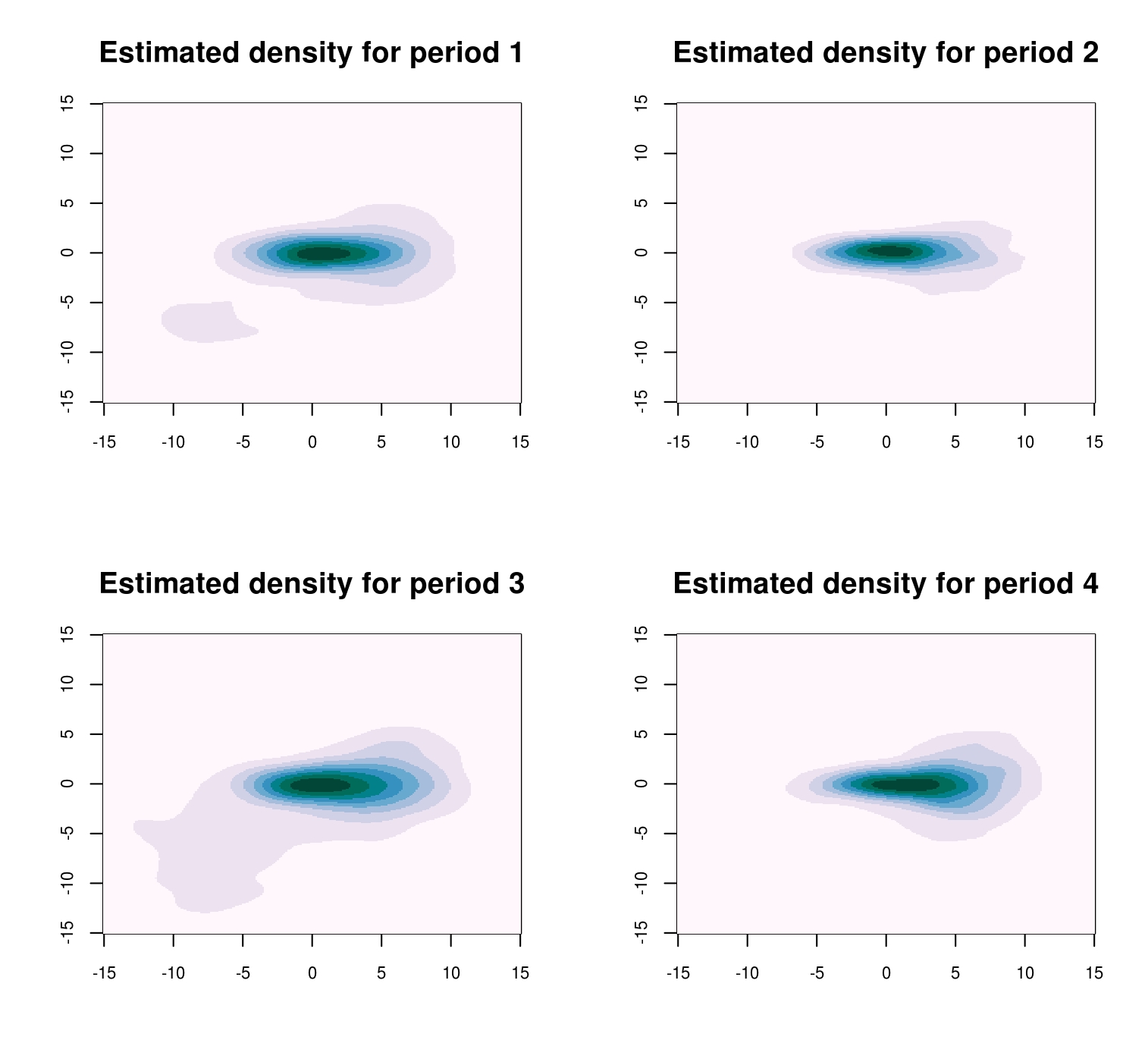}}
	\caption{Estimated density plots for the estimated segmentation provided by the E-Agglo procedure when applied to the spatio-temporal 
point process in Section~\ref{stpp}.}
	\label{estStpp}
\end{figure}
\end{center}

\section{Real data}\label{RealData}

In this section we analyze the results obtained by applying the E-Divisive and E-Agglo methods to two real datasets. 
We first apply our procedures to the micro-array aCGH data from \cite{Vert:2011}. In this dataset we are provided with 
records of the copy-number variations for multiple individuals. Next we examine a set of financial time series. 
For this we consider weekly log returns of the companies which compose the Dow Jones Industrial Average.\par

\subsubsection{Micro-array data}
This dataset consists of micro-array data for 57 different individuals with a bladder tumor. Since all individuals have the same disease, 
we would expect the change point locations to be almost identical on each micro-array set. In this setting, a change point would correspond to a change in copy-number, 
which is assumed to be constant within each segment. The Group Fused Lasso (GFL) approach taken by \cite{Vert:2011} is well suited for this task since 
it is designed to detect changes in mean. We compare the results of our E-Divisive and E-Agglo approaches, when using $\alpha=2$, to those obtained 
by the GFL. In addition, we also consider another nonparametric change point procedure which is able to detect changes in both mean and variability, called 
MultiRank \citep{Fong:2011}.\par

The original dataset from \cite{Vert:2011} contained missing values, and thus our procedure could not be directly applied. Therefore, we removed all 
individuals for which more than 7\% of the values were missing. The remaining missing values we replaced by the average of their neighboring values. 
After performing this cleaning process, we were left with a sample of $d=43$ individuals and size $T=2215$. This dataset can be obtained through the following \proglang{R} commands;
\begin{Schunk}
\begin{Sinput}
R> library("ecp")
R> data("ACGH")
R> acghData = ACGH$data
\end{Sinput}
\end{Schunk}
When applied to the full 43 dimensional series, the GFL procedure estimated 14 change points and the MultiRank procedure estimated 43. When using $\alpha=2,$ the 
E-Divisive procedure estimated 86 change points and the E-Agglo estimated 28.\par

Figures~\ref{person10} and~\ref{person15} provide the results of applying the various methods to a subsample of two individuals 
(persons 10 and 15). The E-Divisive procedure was run with \code{min.size=15}, and \code{R=499}, and the initial segmentation provided to the E-Agglo method 
consisted of equally sized segmens of length 15.
The marginal series are plotted, and the dashed lines are the estimated change point locations.\par

\begin{figure}[h!]
	\centerline{\includegraphics{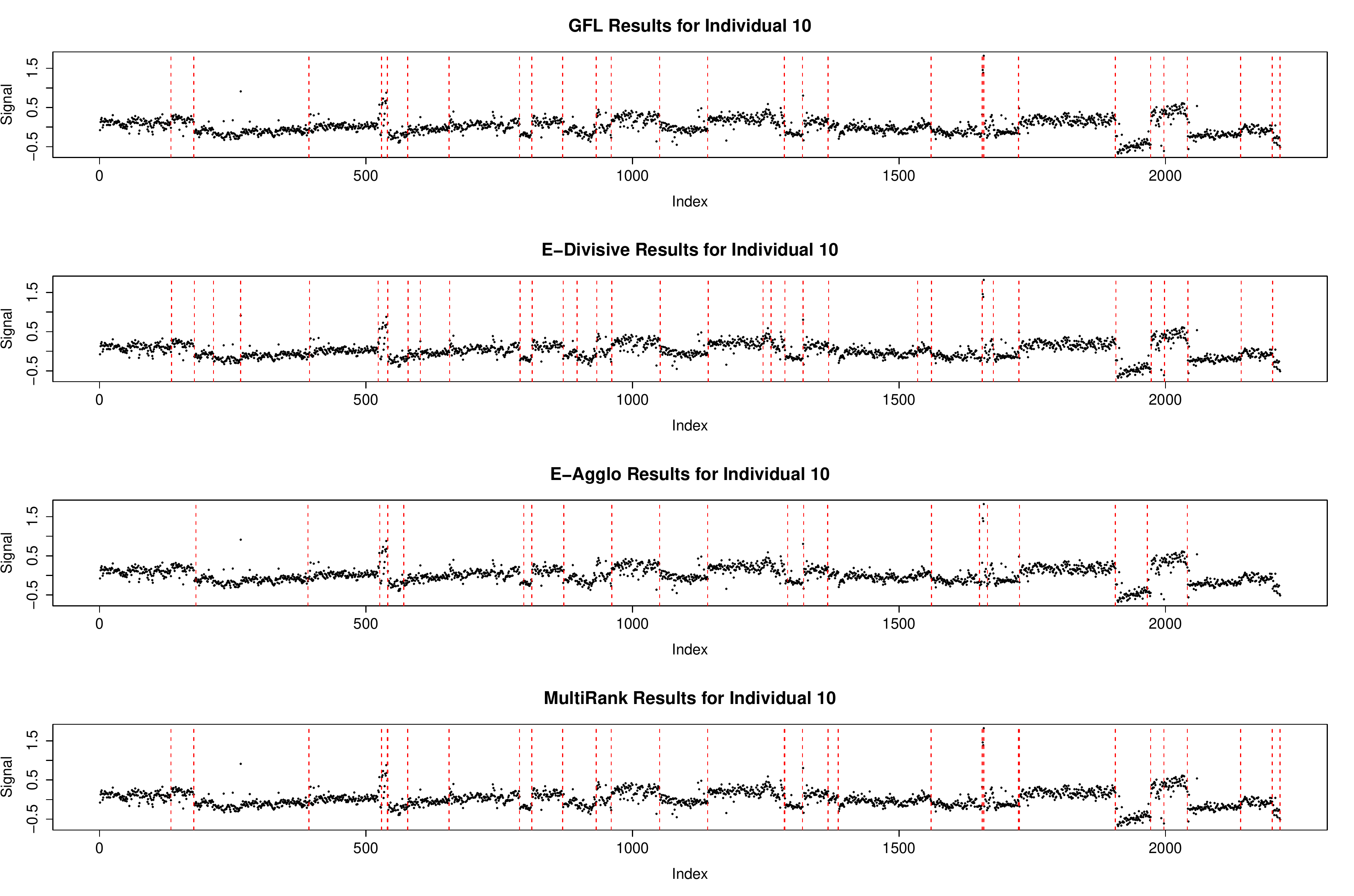}}
	\caption{The aCGH data for individual 10. Estimated change point locations are indicated by dashed vertical lines.}
	\label{person10}
\end{figure}

\begin{figure}[h!]
	\centerline{\includegraphics{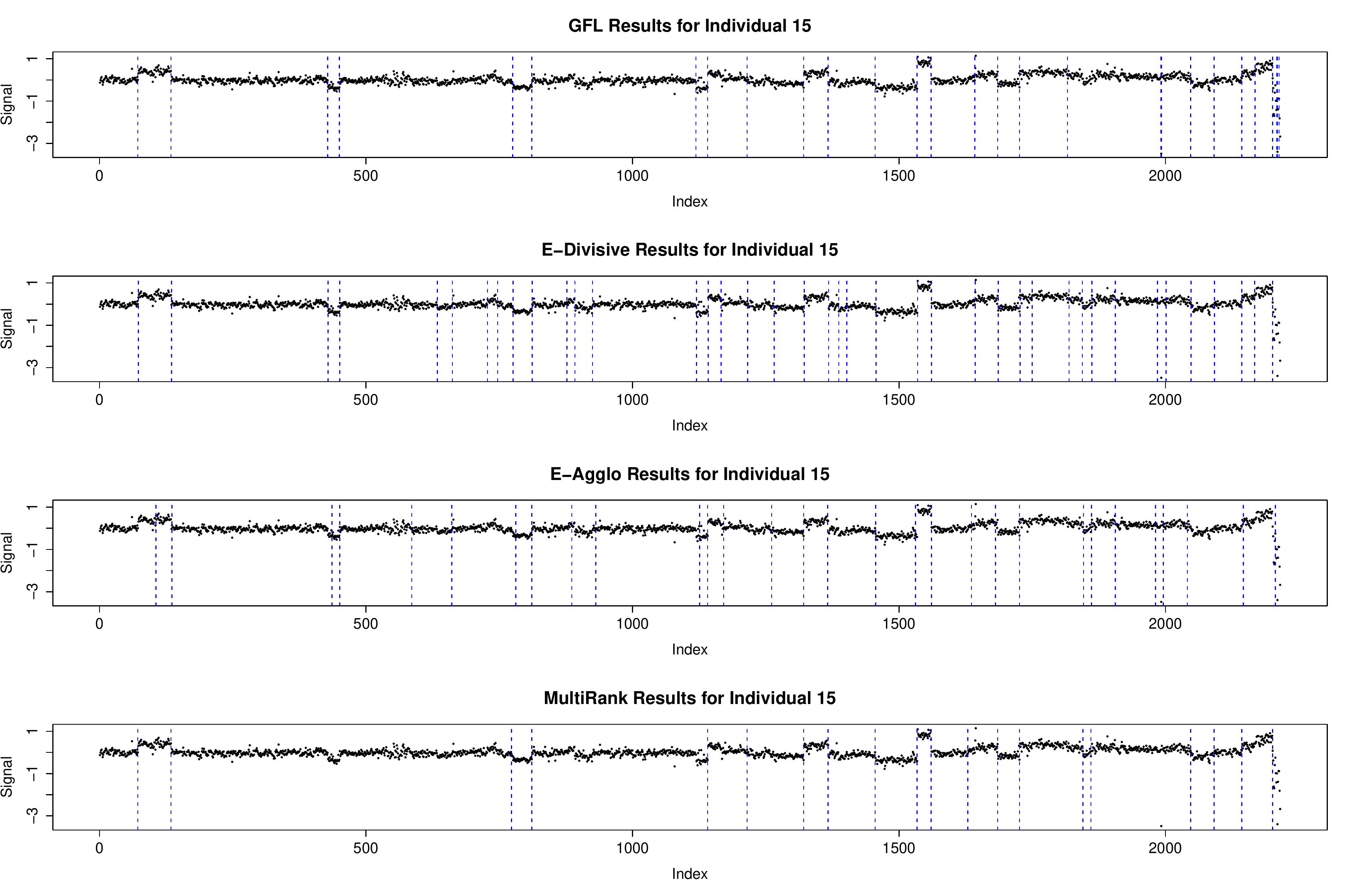}}
	\caption{The aCGH data for individual 15. Estimated change point locations are indicated by dashed vertical lines.}
	\label{person15}
\end{figure}
Looking at the returned estimated change point locations for the full 43-dimensional series we notice that both the E-Divisive and E-Agglo methods identified all of the 
change points returned by the GFL procedure. Further examination also shows that in addition to those change points found by the GFL procedure, the E-Divisive procedure 
also identified changes in the means of the marginal series.  However, if we examine the first 14 to 20 change points estimated by the E-Divisive procedure we observe 
that they are those obtained by the GFL approach. This phenomenon however, does not appear when looking at the results from the E-Agglo procedure. Intuitively this is 
due to the fact that we must provide an initial segmentation of the series, which places stronger limitations on possible change point locaitons, than does specifying 
a minimum segment size.

\subsubsection{Financial data}
Next we consider weekly log returns for the companies which compose the Dow Jones Industrial Average (DJIA). The time period under consideration is 
April 1990 to January 2012, thus providing us with $T=1139$ observations. Since the time series for Kraft Foods Inc.\ does not span this entire period, it is 
not included in our analysis. This dataset is accessible by running \code{data("DJIA")}.\par

When applied to the 29 dimensional series, the E-Divisive method identified change points at 7/13/98, 3/24/03, 9/15/08, and 5/11/09. The change 
points at 5/11/09 and 9/15/08 correspond to the release of the Supervisory Capital Asset Management program results, and the Lehman Brothers bankruptcy 
filing, respectively. If we initially segment the dataset into segments of length 30 and apply the E-Agglo procedure, we identify change points at 
1/3/00, 11/18/02, 8/18/08, and 3/16/09. The change points at 1/3/00 and 3/16/09 correspond to the passing of the Gramm-Leach-Bliley Act and the American 
Recovery and Reinvestment Act respectively.\par

For comparison we also considered the univariate time series for the DJIA Index weekly log returns. In this setting, the E-Divisive method identified change points 
at 10/21/96, 3/31/03, 10/15/07, and 3/9/09. While the E-Agglo method identified change points at 8/18/08 and 3/16/09. Once again, some of these change points correspond 
to major financial events. The change point at 3/9/09 correspond to Moody's rating agency threatening to downgrade Wells Fargo \& Co., JP Morgan Chase \& Co., and 
Bank of America Corp. The 10/15/07 change point is located around the time of the financial meltdown caused by subprime mortgages. In both the univariate and multivariate 
cases the change point in March 2003 is around the time of the 2003 U.S. invasion of Iraq. A plot of the DJIA weekly log returns is provided in Figure~\ref{dowfig} 
along with the locations of the estimated change points by the E-Divisive method.\par

For the E-Divisive method, the set of change points obtained from the univariate and multivariate analysis closely correspond to the same events. However, in the case of the 
E-Agglo method, the multivariate analysis is able to identify significant events that were not able to be detected from the univariate series. For this reason, we would argue 
that regardless of the method being used, it is recommended that multivariate analysis be performed.

\begin{figure}[ht]
	\centerline{\includegraphics{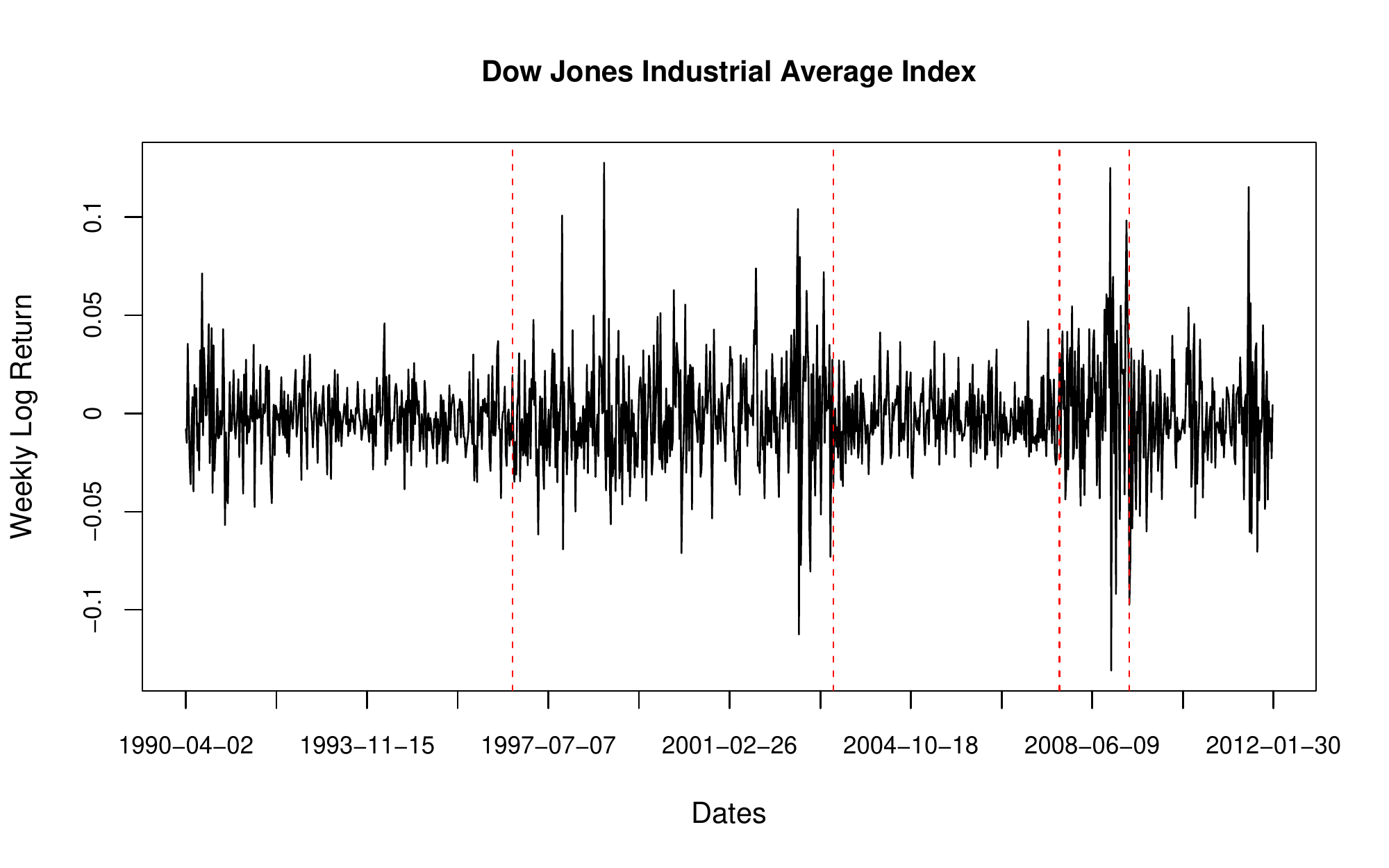}}
	\caption{Weekly log returns for the Dow Jones Industrial Average index from April 1990 to January 2012. The dashed vertical lines indicate the locations of 
estimated change points. The estimated change points are located at 10/21/96, 3/31/03, 10/15/07, and 3/9/09.}
	\label{dowfig}
\end{figure}

\section{Performance analysis}
To compare the performance of different change point methods we used the Rand Index \citep{Rand:1971} as well as Morey and Agresti's Adjusted Rand 
Index \citep{Morey:1984}. These indices provide a measure of similarity between two different segmentations of the same set of observations.\par

The Rand Index evaluates similarity by examining the segment membership of \emph{pairs} of observations. A shortcoming of the Rand Index is that it 
does not measure departure from a given baseline model, thus making it difficult to compare two different 
estimated segmentations. The hypergeometric model is a popular choice for the baseline, and is used by \cite{Hubert:1985} and \cite{Fowlkes:1983}.\par

In our simulation study the Rand and 
Adjusted Rand Indices are determined by comparing the segmentation created by a change point procedure and the true segmentation. We 
compare the performance of our E-Divisive procedure against that of our E-Agglo. The results of the 
simulations are provided in Tables~\ref{sim-all},~\ref{sim-tail} and~\ref{sim-bi-all}. Tables~\ref{sim-all} and~\ref{sim-tail} provide the results for simulations with univariate 
time series, while Table~\ref{sim-bi-all} provides the results for the multivariate time series. In these tables, average Rand Index along with standard 
errors are reported for 1000 simulations. Although not reported, similar results are obtained for the average Adjusted Rand Index.\par

Both the Rand Index and Adjusted Rand Index can be easily obtained through the use of the \code{adjustedRand} function in the \pkg{clues} 
package \citep{Chang:2010}. If \code{U} and \code{V} are membership vectors for two different segmentations of the data, then the required 
index values are obtained as follows,

\begin{Schunk}
\begin{Sinput}
R> library(clues)
R> RAND <- adjustedRand(U,V)
\end{Sinput}
\end{Schunk}
The Rand Index is stored in \code{RAND[1]}, while \code{RAND[2]} and \code{RAND[3]} store various Adjusted Rand indices. These Adjusted 
Rand indices make different assumptions on the baseline model, and thus arrive at different values for the expected Rand index.

\begin{table}[ht!]
\begin{center}
{ \small 
\begin{tabular}{|c|| c| l| l|| c| l| l|}
\hline
  \multicolumn{1}{|c||}{}  & \multicolumn{3}{c||}{{\small\textbf{Change in Mean}}} & \multicolumn{3}{c|}{{\small\textbf{Change in Variance}}}\\
  \hline
  $T$ & $\mu$ & \multicolumn{1}{c|}{E-Divisive} &  \multicolumn{1}{c||}{E-Agglo} & $\sigma^2$ &  \multicolumn{1}{c|}{E-Divisive} & E-Agglo\\
  \noalign{\hrule height 2pt}
\multirow{3}{*}{150}  &   1   & $0.950_{0.001}$ & $0.964_{0.004}$ &           2   & $0.907_{0.003}$ &$0.914_{0.012}$\\
\cline{2-7}          &   2   & $0.992_{4.6\e{-4}}$ & $0.991_{0.001}$  &   5   & $0.973_{0.001}$ &$0.961_{0.002}$ \\
\cline{2-7}          &   4   & $1.000_{3.7\e{-5}}$ & $1.000_{0.000}$ &   10  & $0.987_{7.1\e{-4}}$ &$0.978_{0.002}$ \\
\hline\hline 
\multirow{3}{*}{300}  &   1   & $0.972_{9.1\e{-4}}$   & $0.953_{0.002}$  &   2    & $0.929_{0.003}$&$0.948_{0.021}$ \\
\cline{2-7}          &   2   & $0.996_{2.2\e{-4}}$   & $0.994_{6.4\e{-4}}$  &   5    & $0.990_{5.1\e{-4}}$&$0.976_{0.001}$ \\
\cline{2-7}          &   4   & $1.000_{1.0\e{-5}}$   & $1.000_{0.000}$  &   10   & $0.994_{3.2\e{-4}}$&$0.988_{8.9\e{-4}}$\\
\hline\hline  
\multirow{3}{*}{600}  &   1   & $0.987_{1.5\e{-5}}$ & $0.970_{0.001}$ &   2   & $0.968_{0.001}$ &$ 0.551_{2.3\e{-4}}$ \\
\cline{2-7}          &   2   & $0.998_{3.9\e{-6}}$ & $0.997_{3.0\e{-4}}$ &   5   & $0.995_{2.2\e{-4}}$ &$0.983_{8.7\e{4}}$\\
\cline{2-7}          &   4   & $1.000_{3.1\e{-7}}$ & $1.000_{0.000}$ &   10  & $0.998_{1.5\e{-4}}$ &$0.992_{5.5\e{-4}}$ \\
\hline   
\end{tabular} }
\caption{\label{sim-all} Average Rand Index and standard errors from 1,000 simulations for the E-Divisive and E-Agglo methods. 
Each sample has $T = 150, 300 \;\mathrm{or}\; 600$ observations, consisting of three equally sized clusters, with distributions $N(0,1), G, N(0,1)$, respectively. 
For changes in mean $G \equiv N(\mu,1)$, with $\mu = 1, 2,$ and $4$;
for changes in variance $G \equiv N(0,\sigma^2)$, with $\sigma^2 = 2, 5,$ and $10$.
}
\end{center}
\end{table}
\begin{table}[ht!]
\begin{center}
{ \small 
\begin{tabular}{|c|| c| c| c|}
\hline
  \multicolumn{1}{|c||}{}  & \multicolumn{3}{c|}{{\small\textbf{Change in Tail}}} \\
  \hline
  $T$ & $\nu$ & E-Divisive & E-Agglo\\
  \noalign{\hrule height 2pt}
\multirow{3}{*}{150}  &   16  & $0.835_{0.017}$&$0.544_{6.1\e{-4}}$\\
\cline{2-4}          &   8   & $0.836_{0.020}$&$0.543_{5.9\e{-4}}$  \\
\cline{2-4}          &   2   & $0.841_{0.011}$&$0.545_{7.5\e{-4}}$\\
\hline\hline 
\multirow{3}{*}{300}  &   16  & $0.791_{0.015}$&$0.552_{2.1\e{-4}} $ \\
\cline{2-4}          &   8   & $0.729_{0.018}$&$0.551_{2.2\e{-4}} $\\
\cline{2-4}          &   2   & $0.815_{0.006}$&$ 0.551_{2.3\e{-4}}$\\
\hline\hline  
\multirow{3}{*}{600}  &   16  & $0.735_{0.019}$&$0.552_{2.1\e{-4}} $ \\
\cline{2-4}          &   8   & $0.743_{0.025}$&$ 0.551_{2.2\e{-4}}$\\
\cline{2-4}          &   2   & $0.817_{0.006}$&$ 0.552_{2.3\e{-4}}$\\
\hline   
\end{tabular} }
\caption{\label{sim-tail} Average Rand Index and standard errors from 1,000 simulations for the E-Divisive and E-Agglo methods. 
Each sample has $T = 150, 300 \;\mathrm{or}\; 600$ observations, consisting of three equally sized clusters, with distributions $N(0,1), G, N(0,1)$, respectively. 
For the changes in tail shape $G \equiv t_{\nu}(0,1)$, with $\nu = 16, 8,$ and $2$.
}
\end{center}
\end{table}

\begin{table}[ht!]
\begin{center}
{ \small 
\begin{tabular}{|c|| l| l| l|| c| l| l|}
\hline
  \multicolumn{1}{|c||}{}  & \multicolumn{3}{c||}{{\small\textbf{Change in Mean}}} & \multicolumn{3}{c|}{{\small\textbf{Change in Correlation}}} \\
  \hline
  $T$ & $\mu$ & E-Divisive & E-Agglo& $\rho$ & E-Divisive& E-Agglo\\
  \noalign{\hrule height 2pt}
\multirow{3}{*}{300}  &   1   & $0.987_{4.7\e{-4}}$ & $0.978_{0.001}$ &   0.5   & $0.712_{0.018}$ & $ 0.551_{2.5\e{-4}}$\\
\cline{2-7}                &     2 & $0.992_{8.9\e{-5}}$& $0.999_{2.4\e{4}}$&   0.7   & $0.758_{0.021}$ & $ 0.552_{2.4\e{-4}}$\\
\cline{2-7}                &    3   & $1.000_{1.3\e{-5}}$& $1.000_{0.000}$&   0.9   & $0.769_{0.017}$ & $ 0.550_{3.1\e{-4}}$\\
\hline\hline   
\multirow{3}{*}{600}  &   1   & $0.994_{2.2\e{-4}}$ & $0.986_{8.6\e{4}}$&   0.5   & $0.652_{0.022}$ & $ 0.553_{1.4\e{-4}}$  \\
\cline{2-7}                &     2 & $1.000_{4.3\e{-5}}$ & $0.999_{1.5\e{-4}}$&   0.7   & $0.650_{0.017}$ & $ 0.553_{1.5\e{-4}}$    \\
\cline{2-7}                &    3   & $1.000_{3.3\e{-6}}$ & $1.000_{0.000}$ &   0.9   & $0.806_{0.019}$ & $ 0.553_{1.8\e{-4}}$ \\
\hline\hline   
\multirow{3}{*}{900}  &   1   & $0.996_{1.6\e{-4}}$ & $0.991_{6.0\e{-4}}$&   0.5   & $0.658_{0.024}$ & $ 0.554_{9.9\e{-5}}$\\
\cline{2-7}                &     2 &$1.000_{3.0\e{-5}}$ & $1.000_{7.3\e{-5}}$&   0.7   & $0.633_{0.022}$ & $ 0.554_{1.1\e{-4}}$\\
\cline{2-7}                &    3   & $1.000_{5.2\e{-6}}$ & $1.000_{2.2\e{-5}}$&   0.9   & $0.958_{0.004}$ & $ 0.553_{1.3\e{-4}}$\\
\hline
\end{tabular} }
\caption{\label{sim-bi-all}
Average Rand Index and standard errors from 1,000 simulations for the E-Divisive and E-Agglo methods, when applied to multivariate time series with $d=2$. 
Each sample has $T = 150, 300 \;\mathrm{or}\; 600$ observations, consisting of three equally sized clusters, with distributions $N_2(0,I), G, N_2(0,I)$, respectively. 
For changes in mean $G \equiv N_2(\mu,I)$, with $\mu = (1,1)^\top, (2,2)^\top,$ and $(3,3)^\top$;
for changes in correlation $G \equiv N(0,\Sigma_{\rho})$, in which the diagonal elements of $\Sigma_{\rho}$ are $1$ and the off-diagonal are $\rho$, with $\rho = 0.5, 0.7,$ and $0.9$.
}
\end{center}
\end{table}

\section{Conclusion}\label{conclusion}
The \pkg{ecp} package is able to perform nonparametric change point analysis of multivariate data. The package provides 
two primary methods for performing analysis, each of which is able to determine the number of change points without user input. The only necessary 
user-provided parameter, apart from the data itself, is the choice of $\alpha$. 
If $\alpha$ is selected to lie in the interval $(0,2),$ then the methods provided by this package are able to detect 
\emph{any} type of distributional change within the observed series, provided that the absolute $\alpha$th moments exists.\par

The E-Divisive method sequentially tests the statistical significance of each change point estimate given the previously estimated change locations, 
while the E-Agglo method proceeds by optimizing a goodness-of-fit statistic. For this reason, we prefer to use the E-Divisive method, even though 
its running time is output-sensitive and depends on the number of estimated change points.\par
\newpage
Through the provided examples, applications to real data, and simulations \citep{James:2012}, we observe that the E-Divisive approach 
obtains reasonable estimates for the locations of change points. Currently both the E-Divisive and E-Agglo 
methods have running times that are quadratic relative to the size of the time series. Future version of this package will attempt to reduce this to 
a linear relationship, or provide methods that can be used to quickly provide approximations.

\bibliographystyle{jss}
\nocite{MASSpkg}
\nocite{combinatpkg}
\nocite{Rpkg}
\bibliography{ecp}

\appendix
\section{Appendix}\label{apx}
This appendix provides additional details about the implementation of both the E-Divisive and E-Agglo methods in the \pkg{ecp} 
package.

\subsection{Divisive outline}
The E-Divisive method estimates change points with a bisection approach. In Algorithms~\ref{divOutline} and ~\ref{singleCp}, 
segment $C_i$ contains all observations in time interval $[\ell_i,r_i)$. Algorithm~\ref{singleCp} demonstrates the procedure used to identify 
a single change point. The computational time to maximize over $(\tau,\kappa)$ is reduced to ${\mathcal O}(T^2)$ by using memoization. 
Memoization also allows Algorithm~\ref{singleCp} to execute its for loop at most twice. The permutation test is outlined by Algorithm~\ref{permTest}. 
When given the segmentation $C$, a permutation is only allowed to reorder observations 
so that they remain within their original segments.

\begin{algorithm}
\SetKwInOut{Input}{Inputs}
\SetKwInOut{Output}{Output}
\SetKwFor{For}{for}{}{endfor}
\SetKwFor{While}{while}{}{endwhile}
\caption{Outline of the divisive procedure.}

\Input{Time series $Z$, significance level $p_0$, minimum segment size $m$, the maximum number of permutations for the permutation test $R$, the uniform 
resampling error bound $eps$, epsilon spending rate $h$, and $\alpha\in(0,2]$.}
\Output{A segmentation of the time series.}
\BlankLine
Create distance matrix $Z_{ij}^\alpha=|Z_i-Z_j|^\alpha${\color{white}\;}
\While{Have not found a statisticaly insignificant change point}{
  Estimate next most likely change point location{\color{white}\;}
  Test estimated change point for statistical significance{\color{white}\;}
  \If{Change point is statistically significant}{
    Update the segmentation{\color{white}\;}
  }
}
\Return Final segmentation
\label{divOutline}
\end{algorithm}

\begin{algorithm}
\SetKwInOut{Input}{Inputs}
\SetKwInOut{Output}{Output}
\SetKwFor{For}{for}{}{endfor}
\caption{Outline of procedure to locate a single change point.}

\Input{Segmentation $C$, distance matrix $D$, minimum segment size m.}
\Output{A triple $(x,y,z)$ containing the following information: a segment identifier, a distance within a segment, a weighed sample divergence. }
\BlankLine
best = $-\infty${\color{white}\;}
loc = 0{\color{white}\;}
\For{Segments $C_i\in C$}{
  $A$ = Within distance for $[\ell_i,\ell_i+m)${\color{white}\;}
  \For{$\kappa\in\{\ell_i+m+2,\dots,r_i+1\}$}{
    Calculate and store between and within distances for currenct choice of $\kappa${\color{white}\;}
    Calculate test statistic{\color{white}\;}
    \If{Test statistic $\ge$ best}{
      Update best{\color{white}\;}
      Update loc to m{\color{white}\;}
    }
  }
  \For{$\tau\in\{\ell_i+m+1,\dots,r_i-m\}$}{
    Update within distance for left segment{\color{white}\;}
    \For{$\kappa\in\{\tau+m+1,\dots,r_i+1\}$}{
      Update remaining between and within distances for current choice of $\kappa${\color{white}\;}
      Calcualte test statistic{\color{white}\;}
      \If{Test statistic $\ge$ best}{
	Update best{\color{white}\;}
	Update loc to $\tau${\color{white}\;}
      }
    }
  }
}
\Return Which segment to divide, loc, and best
\label{singleCp}
\end{algorithm}

\begin{algorithm}
\SetKwInOut{Input}{Inputs}
\SetKwInOut{Output}{Output}
\SetKwFor{For}{for}{}{endfor}
\caption{Outline of the permutation test.}

\Input{Distance matrix $D$, observed test statistic $\nu$, maximum number of permutations $R$, uniform resampling error bound $eps$, epsilon spending rate $h$, segmentation $C$, 
minimum segment size m.}
\Output{An approximate $p$~value.}
\BlankLine
over = 1{\color{white}\;}
\For{$i\in\{1,2,\dots,R\}$}{
	Permute rows and columns of $D$ based on the segmentation $C$ to create $D'${\color{white}\;} 
	Obtain test statistic for permuted observations{\color{white}\;}
	\If{Permuted test statistic $\ge$ observed test statistic}{
		over = over + 1{\color{white}\;}
	}
	\If{An early termination condition is satisfied}{
		\Return over/(i+1)
	}
}
\Return over/(R+1)
\label{permTest}
\end{algorithm}

\pagebreak

\subsection{Agglomerative outline}
The E-Agglo method estimates change point by maximizing the goodness-of-fit statistic given by Equation~\ref{gof}. 
The method must be provided an initial segmentation of the series. Segments are then merged in order to maximize the 
goodness-of-fit statistic. As segments are merged, their between-within distances also need to be 
updated. The following result due to \cite{Rizzo:2005} greatly reduces the computational time necessary to perform these updates.
\begin{lem}
	Suppose that $C_1, C_2,$ and $C_3$ are disjoint segments with respective sizes $m_1, m_2,$ and $m_3$. Then if $C_1$ and $C_2$ are merged 
to form the segment $C_1\cup C_2$,
{\small
$$\widehat{\cal E}(C_1\cup C_2,C_3;\alpha)=\frac{m_1+m_3}{m_1+m_2+m_3}\widehat{\cal E}(C_1,C_3;\alpha)+
\frac{m_2+m_3}{m_1+m_2+m_3}\widehat{\cal E}(C_2,C_3;\alpha)-
\frac{m_3}{m_1+m_2+m_3}\widehat{\cal E}(C_1,C_2;\alpha).$$
}
\end{lem}
Algorithm~\ref{aggloOutline} is an outline for the agglomerative procedure. In this outline $C_{i+k}\ (C_{i-k})$ is the segment that is $k$ 
segments to the right (left) of $C_i$.

\begin{algorithm}[ht]
\SetKwInOut{Input}{Inputs}
\SetKwInOut{Output}{Output}
\SetKwFor{For}{for}{}{endfor}
\caption{Outline of the agglomerative procedure.}

\Input{An initial segmentation C, a time series $Z$, a penalty function $f(\vec\tau)$, and $\alpha\in(0,2]$.}
\Output{A segmentation of the time series.}
\BlankLine
Create distance matrix $D_{i,j}=\widehat{\mathcal E}(C_i,C_j;\alpha)${\color{white}\;}
Obtain initial penalized goodness-of-fit (gof) statistic{\color{white}\;}
\For{$K\in\{N,N+1,\dots,2N-3\}$}{
	Merge best candidate segments{\color{white}\;}
	Update current gof{\color{white}\;}
	\If{Current gof $\ge$ largest gof so far}{
	  Update largest gof
	  }
}
Penalize the sequence of obtained gof statistics{\color{white}\;}
Choose best segmentation based on penalized gof statistics{\color{white}\;}
\Return Best segmentation
\label{aggloOutline}
\end{algorithm}

\end{document}